\documentclass[12pt]{amsart}
\usepackage{amsfonts}
\usepackage{ifthen}
\usepackage{verbatim}
\usepackage{amsthm}
\usepackage{amsmath}
\usepackage{amssymb}
\usepackage{graphicx}
\usepackage{subfigure}
\usepackage{amscd,amssymb,amsthm}
\usepackage{color,xcolor}

\newcounter{minutes}
\divide\time by 60
\newcounter{hours}
\multiply\time by 60 \addtocounter{minutes}{-\time}

\title{Moment inequalities of the second and third order}

\author{Slavko Simi\'c}

\address{ Mathematical Institute SANU, Kneza Mihaila 36, 11000
Belgrade, Serbia} \email{ ssimic@turing.mi.sanu.ac.rs}

\keywords{moment inequalities; Jensen's inequality; positive
semi-definite form; relative divergence.} \subjclass{60E15}

\newtheorem{lemma}[equation]{Lemma}

\newcommand{\beq}{\begin{equation}}
\newcommand{\eeq}{\end{equation}}

\numberwithin{equation}{section}

\begin{document}

\def\thefootnote{}
\footnotetext{ \texttt{\tiny File:~\jobname .tex,
          printed: \number\year-\number\month-\number\day,
          \thehours.\ifnum\theminutes<10{0}\fi\theminutes}
} \makeatletter\def\thefootnote{\@arabic\c@footnote}\makeatother


\maketitle

\begin{abstract}
In this paper we give refinements of some convex and log-convex
moment inequalities of the second and third order using a special
kind of an algebraic positive semi-definite form. An open problem
concerning eight parameter refinement of the third order is also
stated with some applications in Information Theory concerning
relative divergence of type $s$.
\end{abstract}


\section{Introduction}

\vspace{0.5cm}

Inequalities for moments of s-th order $EX^s$, of a probability
law with support on $\mathbb R^+ $, are of fundamental interest in
Probability Theory. Most known of them all is Jensen's Moment
Inequality, given by

$$
\begin{array}{ll}
EX^s\ge (EX)^s, & s\in (-\infty,0)\cup (1, +\infty);\\
EX^s\le (EX)^s, & s\in (0,1).
\end{array}.
$$

\vspace{0.5cm}

Or, more generally,

\vspace{0.5cm}

{\bf Theorem A} {\it If $F$ is a convex function on an interval
$I\subset\mathbb R$, then

$$
E(F(X))\ge F(EX)
$$

for any probability law with support on $I$.}

 \vspace{0.5cm}

 The topic of research in this article is a difference of
 moments, introduced in the following way.

For an arbitrary probability law with support on $(0, \infty)$,
define

$$
\lambda_s=\lambda_s(X):=\left\{
\begin{array}{ll}
(EX^s-(EX)^s)/s(s-1), & s\neq 0, 1;\\
\log(EX)-E(\log X), & s=0;\\
E(X\log X)-(EX)\log(EX), & s=1.
\end{array}\right.
$$

\vspace{0.5cm}

Throughout the paper we suppose that moments exist for all $s\in
I\subset\mathbb R$.

\bigskip

 The above Jensen's Moment Inequality yields that
$\lambda_s\ge 0$ for $s\in I$.

\bigskip

It is also known (\cite{s}) that $\lambda_s$ is log-convex (hence
convex) in $s$, that is
$$
\xi(s, t):=\lambda_s\lambda_t-\lambda^2_\frac{s+t}{2}\ge 0; s,
t\in I.\eqno(1)
$$

\vspace{0.5cm}

 Moreover, we proved the following mixture of Jensen-Lyapunov moment
 inequalities (\cite{st}),

\vspace{0.5cm}

 {\bf Theorem B} {\it For any $-\infty<r<s<t<+\infty$, we have}
 $$
(\lambda_s)^{t-r}\le (\lambda_r)^{t-s}(\lambda_t)^{s-r}. \eqno (2)
 $$

\vspace{0.5cm}

Note that for $s\in \mathbf N; r, t\in \mathbf {2N}$, the relation
(2) is valid for arbitrary probability distributions with support
on $(-\infty, +\infty)$, (\cite{ss}).

\vspace{0.5cm}

\section{ Results}

\vspace{0.5cm}

Our aim in this article is to give some refinements of the above
moment inequalities.

\vspace{0.5cm}

Namely, the fact that $\lambda_s$ is convex i.e.
$\tau(s,t):=\lambda_s-2\lambda_\frac{s+t}{2}+\lambda_t\ge 0$, is
improved to the following four parameter inequality of second
order.

\vspace{0.5cm}

{\bf Theorem 1} {\it For arbitrary $s,t,u,v\in I$, we have}

$$
\tau(s,t)\tau(u,v)\ge
[\tau(\frac{s+u}{2},\frac{t+v}{2})-\tau(\frac{s+v}{2},\frac{t+u}{2})]^2.
$$

\vspace{0.5cm}

As a corollary, we get the next interesting assertion.

\vspace{0.5cm}

{\bf Corollary 1} {\it For fixed $a\in\mathbb R$, the function
$\mu_a(t):=\lambda_t-2\lambda_{t+a/2}+\lambda_{t+a}$ is log-convex
in $t$.}

\vspace{0.5cm}

Also,

\vspace{0.5cm}

{\bf Theorem 2}  {\it For arbitrary $p,q\ge 0$, the function
$w(s,t):=p^2\lambda_s+2pq \lambda_\frac{s+t}{2}+q^2\lambda_t$ is
log-convex in $s,t\in I$, that is

$$
w(s,t) w(u,v)\ge [w(\frac{s+u}{2}, \frac{t+v}{2})]^2.
$$}

\vspace{0.5cm}

The inequality (1) can be improved to the next assertion of second
order.

\vspace{0.5cm}

{\bf Theorem 3} {\it Denote
$$
\phi(s,t;u,v):=\xi(s,t)\xi(u,v)-[\xi(\frac{s+u}{2},\frac{t+v}{2})-\xi(\frac{s+v}{2},\frac{t+u}{2})]^2.
$$

Then the inequality
$$
\phi(s,t;u,v)\ge 0,
$$

holds for each $s,t,u,v\in I$.}

\bigskip

The above inequality is very sharp. For instance, taking the
discrete probability law with masses $x$ and $y$ and assigned
probabilities $p$ and $q$, a calculation shows that
$$
\phi(2,4;2,6)=C(p,q)(x-y)^{14},
$$

where

$$
C(p,q)=\frac{(pq)^4}{5529600}[35+11(p-q)^2+17(p-q)^4+(p-q)^6].
$$

\vspace{0.5cm}

As a consequence we obtain the next log-convexity assertion.

\vspace{0.5cm}

{\bf Corollary 2} {\it For some $a\in \mathbb R$, the function
$\sigma(x)=\sigma_a(x):=\xi(x, x+a)$ is log-convex on I.}

\vspace{0.5cm}

Based on a plenty of calculated examples, we conclude that
refinements of the third order are also possible.

 Its general form is given by the
following 8-parameters hypothesis.

\vspace{0.5cm}

{\bf Conjecture 1} {\it The inequality

$$
\phi(r_1,s_1;u_1,v_1)\phi(r_2,s_2;u_2,v_2)\ge
$$
$$
[(\xi(\frac{r_1+r_2}{2},\frac{s_1+s_2}{2})-\xi(\frac{r_1+s_2}{2},\frac{s_1+r_2}{2}))(\xi(\frac{u_1+u_2}{2},\frac{v_1+v_2}{2})-\xi(\frac{u_1+v_2}{2},\frac{v_1+u_2}{2}))
$$
$$
-(\xi(\frac{r_1+u_2}{2},\frac{s_1+v_2}{2})-\xi(\frac{r_1+v_2}{2},\frac{s_1+u_2}{2}))(\xi(\frac{u_1+r_2}{2},\frac{v_1+s_2}{2})-\xi(\frac{u_1+s_2}{2},\frac{v_1+r_2}{2}))]^2
$$

holds for arbitrary $r_i, s_i, u_i, v_i\in I, i\in\{1,2\}$?}

\vspace{0.5cm}

We are able to prove Conjecture 1  in some particular cases.

\bigskip

The first one is for $r_1=u_1=r_2, s_1=u_2, v_1=s_2=v_2$.
Therefore, we obtain a 3-parameter refinement of the third order,
given by

\bigskip

 {\bf Theorem 4} {\it  For any $r,s,v\in I$, we have
$$
\phi(r,s; r,v) \phi(r,v; s,v)\ge
$$
$$
\ge [\xi(r,v)(\xi(s,\frac{r+v}{2})-\xi(\frac{r+s}{2},
\frac{s+v}{2}))
$$
$$
+ (\xi(r,\frac{s+v}{2})-\xi(\frac{r+s}{2},
\frac{r+v}{2}))(\xi(v,\frac{r+s}{2})-\xi(\frac{r+v}{2},\frac{s+v}{2}))]^2.
$$}

\vspace{0.5cm}

The second case is conditioned by $r_1=r_2, s_1=c_1=s_2=u_2$. The
4-parameter improvement of the third order follows as

\vspace{0.5cm}

 {\bf Theorem 5} {\it  For any $r,s,u,v\in I$, we
have

 $$
\phi(r,s; s,u) \phi(r,s; s,v)\ge
$$
$$
\ge [\xi(r,s)(\xi(s,\frac{u+v}{2})-\xi(\frac{s+u}{2},
\frac{s+v}{2}))
$$
$$
- (\xi(s,\frac{r+u}{2})-\xi(\frac{r+s}{2},
\frac{s+u}{2}))(\xi(s,\frac{r+v}{2})-\xi(\frac{r+s}{2},\frac{s+v}{2}))]^2.
$$}

\vspace{0.5cm}

Another conjecture is related to the function
$\sigma_a(x):=\xi(x,x+a)$. By the result of Corollary 2, we have
that
$$
\theta_a(x,y):=\sigma_a(x)\sigma_a(y)-\sigma_a^2(\frac{x+y}{2})\ge
0.
$$

It seems that an inequality analogous to the one from Theorem 3
holds in this case.

\vspace{0.5cm}

{\bf Conjecture 2} {\it Is it true that

$$
\theta_a(x,y)\theta_a(z,v)\ge
[\theta_a(\frac{x+z}{2},\frac{y+v}{2})-\theta_a(\frac{x+v}{2},\frac{y+z}{2})]^2?
$$}

\vspace{0.5cm}

\section{Proofs}

\vspace{0.5cm}

In order to prove our results, the following Main Lemma is of
crucial importance.

\begin{lemma}\label{l1}
For arbitrary real numbers $a,b,c,d$ and $r,s,u,v\in I$, the form
$$
\psi(a,b,c,d):=[a^2\lambda_r+2ab\lambda_\frac{r+s}{2}+b^2\lambda_s][c^2\lambda_u+2cd\lambda_\frac{u+v}{2}+d^2\lambda_v]
$$
$$
- [ac\lambda_\frac{r+u}{2}+ad\lambda_\frac{r+v}{2}
+bc\lambda_\frac{s+u}{2} +bd\lambda_\frac{s+v}{2}]^2
$$

is positive semi-definite i.e., $\psi(a,b,c,d)\ge 0$.
\end{lemma}

\begin{proof}
For a variable $x\in\mathbb R^+$ and a parameter $s\in \mathbb R$
consider the function $f_s(x)$, defined as

$$
f_s(x):=\begin{cases} {(x^s-sx+s-1)/s(s-1)} &, s(s-1)\neq 0;\\
              x-\log x-1 &, s=0;\\
              x\log x-x+1 &, s=1.
              \end{cases}
              $$
              \bigskip
Since,
$$
f'_s(x)=\begin{cases} {x^{s-1}-1\over s-1} &, s(s-1)\neq 0;\\
               1-{1\over x} &, s=0;\\
               \log x &, s=1,
               \end{cases}
               $$
and
$$
f''_s(x)=x^{s-2},
$$

it follows that $f_s(x)$ is a twice continuously differentiable
convex function.

\vspace{0.5cm}

Now, for arbitrary parameters $X,Y\in\mathbb R$, denote

$$
F(x):=(a^2f_r(x)+2abf_\frac{r+s}{2}(x)+b^2f_s(x))X^2+2(ac
f_\frac{r+u}{2}(x)+ad f_\frac{r+v}{2}(x) +bc f_\frac{s+u}{2}(x)
+bd f_\frac{s+v}{2}(x))XY
$$
$$
+(c^2f_u(x)+2cdf_\frac{u+v}{2}(x)+d^2f_v(x))Y^2.
$$

Since,

$$
F''(x)=\frac{1}{x^2}[(a^2x^r+2abx^\frac{r+s}{2}+b^2x^s)X^2+2(ac
x^\frac{r+u}{2}+ad x^\frac{r+v}{2} +bc x^\frac{s+u}{2} +bd
x^\frac{s+v}{2})XY
$$
$$
+(c^2x^u+2cdx^\frac{u+v}{2}+d^2x^v)Y^2]
$$
$$
=\frac{1}{x^2}[(ax^\frac{r}{2}+bx^\frac{s}{2})^2X^2+2(ax^\frac{r}{2}+bx^\frac{s}{2})(cx^\frac{u}{2}+dx^\frac{v}{2})XY
+(cx^\frac{u}{2}+dx^\frac{v}{2})^2Y^2]
$$
$$
=\frac{1}{x^2}[(ax^\frac{r}{2}+bx^\frac{s}{2})X+(cx^\frac{u}{2}+dx^\frac{v}{2})Y]^2,
$$

we conclude that $F$ is a convex function for $x>0$.

\bigskip

Applying Theorem A in this case, we obtain

$$
(a^2\lambda_r+2ab\lambda_\frac{r+s}{2}+b^2\lambda_s)X^2+2(ac
\lambda_\frac{r+u}{2}+ad\lambda_\frac{r+v}{2} +bc
\lambda_\frac{s+u}{2} +bd\lambda_\frac{s+v}{2})XY
$$
$$
+(c^2\lambda_u+2cd\lambda_\frac{u+v}{2}+d^2\lambda_v)Y^2\ge 0.
$$

Since the function $\lambda$ is log-convex, note that the terms
$$
\alpha:=a^2\lambda_r+2ab\lambda_\frac{r+s}{2}+b^2\lambda_s
$$

and

$$
\gamma:=c^2\lambda_u+2cd\lambda_\frac{u+v}{2}+d^2\lambda_v
$$

are non-negative for arbitrary $a,b,c,d\in\mathbb R$.

\vspace{0.5cm}

On the other hand, the form
$$
\alpha X^2+2\beta XY+\gamma Y^2, \ \alpha,\gamma\ge 0,
$$

is positive semi-definite if and only if $\alpha\gamma-\beta^2\ge
0$.

The proof of Lemma \ref{l1} readily follows.

\end{proof}

\vspace{0.5cm}

\begin{proof} of Theorem 1.

\vspace{0.5cm}

 Taking $a=c=-b=-d=1$ in Lemma \ref{l1}, we obtain

$$
[\lambda_s-2\lambda_\frac{s+t}{2}+\lambda_t][\lambda_u-2\lambda_\frac{u+v}{2}+\lambda_v]
$$
$$
\ge[\lambda_\frac{s+u}{2}-\lambda_\frac{s+v}{2}+\lambda_\frac{t+v}{2}-\lambda_\frac{t+u}{2}]^2,
$$

 which is equivalent to the result of Theorem 1.

 \vspace{0.5cm}

For $s=t+a, v=u+a$, we get
$$
\mu_a(t)\mu_a(u)\ge\mu_a^2(\frac{t+u}{2}),
$$

i.e. the assertion from Corollary 1.

\end{proof}

\vspace{0.5cm}

\begin{proof} of Theorem 2.

\vspace{0.5cm}

Since $\lambda(s)$ is a convex function, taking $a=c=p; b=d=q;
p,q>0$ in the above lemma, we get

$$
w(s,t)w(u,v)=(p^2\lambda_s+2pq
\lambda_\frac{s+t}{2}+q^2\lambda_t)(p^2\lambda_u+2pq
\lambda_\frac{u+v}{2}+q^2\lambda_v)
$$
$$
\ge [p^2\lambda_\frac{s+u}{2}+pq(\lambda_
\frac{s+v}{2}+\lambda_\frac{t+u}{2})+q^2\lambda_\frac{t+v}{2}]^2
$$
$$
\ge [p^2\lambda_\frac{s+u}{2}+2pq \lambda_
\frac{s+u+t+v}{4}+q^2\lambda_\frac{t+v}{2}]^2 =
w^2(\frac{s+u}{2},\frac{t+v}{2}),
$$

as desired.

\end{proof}

\vspace{0.5cm}

\begin{proof} of Theorem 3.

\vspace{0.5cm}

For complete proof of the assertion of this theorem, we should
preliminary prove that it is valid in the special case $v=s$, that
is

\begin{lemma}\label{l2}
For arbitrary $s,t,u\in I$, we have that

$$
\phi(s,t;s,u):=\xi(s,t)\xi(s,u)-[\xi(s,\frac{t+u}{2})-\xi(\frac{s+u}{2},\frac{s+t}{2})]^2\ge
0.
$$
\end{lemma}

\begin{proof}

Since

$$
\psi(0,1,c,d)=\lambda_s[c^2\lambda_u+2cd\lambda_\frac{u+v}{2}+d^2\lambda_v]-[c\lambda_\frac{s+u}{2}+d\lambda_\frac{s+v}{2}]^2
$$
$$
=[\lambda(s)\lambda(u)-\lambda^2(\frac{s+u}{2})]c^2+2[\lambda(s)\lambda(\frac{u+v}{2})-\lambda(\frac{s+u}{2})\lambda(\frac{s+v}{2})]cd
$$
$$
+[\lambda(s)\lambda(v)-\lambda^2(\frac{s+v}{2})]d^2,
$$

and this quadratic form is positive semi-definite, then
necessarily,

$$
0\le[\lambda_s\lambda_u-\lambda^2_\frac{s+u}{2}][\lambda_s\lambda_v-\lambda^2_\frac{s+v}{2}]
$$
$$
-[\lambda_s\lambda_\frac{u+v}{2}-\lambda_\frac{s+u}{2}\lambda_\frac{s+v}{2}]^2=\phi(s,u;s,v).
$$

\end{proof}

Choose now numbers $a,b,c,d$ in Lemma \ref{l1} such that

$$
a^2=[\lambda_u\lambda^2_\frac{s+v}{2}-2\lambda_\frac{s+u}{2}\lambda_\frac{s+v}{2}\lambda_\frac{u+v}{2}+\lambda_v\lambda^2_\frac{s+u}{2}]/\xi(u,v)
$$
$$
=\frac{1}{\lambda_u}[\lambda^2_\frac{s+u}{2}+(\xi(u,\frac{s+v}{2})-\xi(\frac{s+u}{2},\frac{u+v}{2}))^2/\xi(u,v)];
$$

$$
(b\lambda_s+a\lambda_\frac{r+s}{2})^2=\frac{\xi(r,s)}{\lambda(u)\xi(u,v)}\phi(u,s;u,v);
$$

$$
c=-\lambda_\frac{s+v}{2}/a; \ d=\lambda_\frac{s+u}{2}/a.
$$

\vspace{0.5cm}

A calculation shows that this choice gives

$$
a^2\lambda_r+2ab\lambda_\frac{r+s}{2}+b^2\lambda_s=\frac{1}{\lambda_s}[(b\lambda_s+a\lambda_\frac{r+s}{2})^2+a^2\xi(r,s)]
$$

$$
=\frac{1}{\lambda_s}[\frac{\xi(r,s)}{\lambda_u\xi(u,v)}\phi(u,s;u,v)
+\frac{\xi(r,s)}{\lambda_u}[\lambda^2_\frac{s+u}{2}+\frac{(\xi(u,\frac{s+v}{2})-\xi(\frac{s+u}{2},\frac{u+v}{2}))^2}{\xi(u,v)}]]
$$

$$
=\frac{\xi(r,s)}{\lambda_u\lambda_s\xi(u,v)}[\phi(u,s;u,v)+(\xi(u,\frac{s+v}{2})-\xi(\frac{s+u}{2},\frac{u+v}{2}))^2+\xi(u,v)\lambda^2_\frac{s+u}{2}]
$$

$$
=\frac{\xi(r,s)}{\lambda_u\lambda_s}[\xi(s,u)+\lambda^2_\frac{s+u}{2}]=\xi(r,s).
$$

By the definition of $c$ and $d$, we also get

$$
c^2\lambda_u+2cd\lambda_\frac{u+v}{2}+d^2\lambda_v
$$
$$
=\frac{1}{a^2}[\lambda_u\lambda^2_\frac{s+v}{2}-2\lambda_\frac{s+u}{2}\lambda_\frac{s+v}{2}\lambda_\frac{u+v}{2}+\lambda_v\lambda^2_\frac{s+u}{2}]=\xi(u,v),
$$

and

$$
ac\lambda_\frac{r+u}{2}+ad\lambda_\frac{r+v}{2}
+bc\lambda_\frac{s+u}{2} +bd\lambda_\frac{s+v}{2}
$$
$$
=-\lambda_\frac{r+u}{2}\lambda_\frac{s+v}{2}+\lambda_\frac{s+u}{2}\lambda_\frac{r+v}{2}-(b/a)\lambda_\frac{s+v}{2}\lambda_\frac{s+u}{2}
+(b/a)\lambda_\frac{s+v}{2}\lambda_\frac{s+u}{2}
$$
$$
=\lambda_\frac{s+u}{2}\lambda_\frac{r+v}{2}-\lambda_\frac{r+u}{2}\lambda_\frac{s+v}{2}
=\xi(\frac{s+u}{2},\frac{r+v}{2})-\xi(\frac{r+u}{2},\frac{s+v}{2}).
$$

Therefore, applying Lemma \ref{l1} for the given choice of
parameters, we obtain the result of Theorem 3.

\end{proof}

\vspace{0.5cm}

\begin{proof} of Corollary 2.

\vspace{0.5cm}

Indeed, by Theorem 3 we get

$$
0\le
\phi(x,x+a;y,y+a)=\xi(x,x+a)\xi(y,y+a)-[\xi(\frac{x+y}{2},\frac{x+y}{2}+a)-\xi(\frac{x+y+a}{2},\frac{x+y+a}{2})]^2
$$
$$
=\sigma_a(x)\sigma_a(y)-\sigma_a^2(\frac{x+y}{2}).
$$

\end{proof}

\vspace{0.5cm}

\begin{proof} of Theorem 4.

\vspace{0.5cm}

In the sequel we need the following elementary assertion.

\begin{lemma}\label{l3}

Denote $D:=\alpha\beta-\gamma^2, \
E:=\eta-\frac{1}{\alpha}[\delta^2+\frac{(\alpha\varepsilon-\gamma\delta)^2}{D}]$.

Then the form
$$
F(a,c):=\alpha a^2+\beta c^2+2\gamma ac+2\delta a+2\varepsilon c
+\eta
$$
is positive semi-definite if and only if $\alpha\ge 0, D\ge 0,
E\ge 0$.

\end{lemma}

\begin{proof}

From the identity

$$
F(a,c)=\frac{1}{\alpha}[(\alpha a+\gamma
c+\delta)^2+\frac{1}{D}(Dc+\alpha\varepsilon-\gamma\delta)^2]+E,
$$

the proof easily follows.

\end{proof}

Now, developing the form $\psi(a,b,c,d)$ with $b=d=1, u=r$, we get

$$
0\le \psi(a,1,c,1)=\xi(r,v)
a^2-2[\lambda_r\lambda_\frac{s+v}{2}-\lambda_\frac{r+s}{2}\lambda_\frac{r+v}{2}]ac
+\xi(r,s)c^2
$$
$$
+2[\lambda_v\lambda_\frac{r+s}{2}-\lambda_\frac{v+s}{2}\lambda_\frac{r+v}{2}]a
+2[\lambda_s\lambda_\frac{r+v}{2}-\lambda_\frac{r+s}{2}\lambda_\frac{s+v}{2}]c
+\xi(s,v),
$$

\bigskip

and, applying Lemma \ref{l3} with

$$
\alpha=\xi(r,v); \ \beta=\xi(r,s); \
\gamma=-[\xi(r,\frac{s+v}{2})-\xi(\frac{r+s}{2},\frac{r+v}{2})];
$$
$$
\delta=[\xi(v,\frac{r+s}{2})-\xi(\frac{v+s}{2},\frac{r+v}{2})]; \
\varepsilon=[\xi(s,\frac{r+v}{2})-\xi(\frac{r+s}{2},\frac{s+v}{2})];
 \ \eta=\xi(s,v),.
$$

we obtain the proof since in this case

$$
D=\phi(r,v;r,s); E=.\frac{F}{\xi(r,v)\phi(r,v;r,s)},
$$

where $F$ is exactly

$$
F=\phi(r,s; r,v) \phi(r,v; s,v)
$$
$$
- [\xi(r,v)(\xi(s,\frac{r+v}{2})-\xi(\frac{r+s}{2},
\frac{s+v}{2}))
$$
$$
+ (\xi(r,\frac{s+v}{2})-\xi(\frac{r+s}{2},
\frac{r+v}{2}))(\xi(v,\frac{r+s}{2})-\xi(\frac{r+v}{2},\frac{s+v}{2}))]^2,
$$

as desired.

\end{proof}

{\bf Remark 1} \ {\it The side result $D\ge 0$ yields another
proof of Lemma \ref{l2}}.

\vspace{0.5cm}

\begin{proof} of Theorem 5.

\vspace{0.5cm}

Developing the form $\phi(a,1,c,d)$ in $a$, we get

$$
0\le
\phi(a,1,c,d)=a^2[\lambda_rC-D^2]+2a[\lambda_\frac{r+s}{2}C-DE]+
[\lambda_sC-E^2],
$$

where

$$
C:=\lambda_u c^2+2\lambda_\frac{u+v}{2}cd+\lambda_v d^2;
$$
$$
D:=\lambda_\frac{r+u}{2}c+\lambda_\frac{r+v}{2}d;
$$
$$
E:=\lambda_\frac{s+u}{2}c +\lambda_\frac{s+v}{2}d.
$$

Hence,

$$
[\lambda_rC-D^2][\lambda_sC-E^2]-[\lambda_\frac{r+s}{2}C-DE]^2\ge
0,
$$

which is equivalent to

$$
0\le \xi(r,s)C+2\lambda_\frac{r+s}{2}DE-\lambda_r E^2-\lambda_s
D^2
$$
$$
=\alpha c^2+2\gamma cd +\beta d^2.
$$

Calculating the coefficients in this case, we obtain

$$
\alpha=\frac{1}{\lambda_s}[\xi(r,s)\xi(u,s)-(\lambda_s
\lambda_\frac{r+u}{2}-\lambda_\frac{r+s}{2}\lambda_\frac{s+u}{2})^2]=\frac{\phi(r,s;u,s)}{\lambda_s};
$$
$$
\beta=\frac{1}{\lambda_s}[\xi(r,s)\xi(v,s)-(\lambda_s
\lambda_\frac{r+v}{2}-\lambda_\frac{r+s}{2}\lambda_\frac{s+v}{2})^2]=\frac{\phi(r,s;v,s)}{\lambda_s};
$$
$$
\gamma=\frac{1}{\lambda_s}[\xi(r,s)(\lambda_s\lambda_\frac{u+v}{2}-\lambda_\frac{s+u}{2}\lambda_\frac{s+v}{2})-(\lambda_s
\lambda_\frac{r+u}{2}-\lambda_\frac{r+s}{2}\lambda_\frac{s+u}{2})(\lambda_s
\lambda_\frac{r+v}{2}-\lambda_\frac{r+s}{2}\lambda_\frac{s+v}{2})]
$$
$$
=\frac{1}{\lambda_s}[\xi(r,s)(\xi(s,\frac{u+v}{2})-\xi(\frac{s+u}{2},\frac{s+v}{2}))-(\xi(s
,\frac{r+u}{2})-\xi(\frac{r+s}{2},\frac{s+u}{2}))(\xi(s,
\frac{r+v}{2})-\xi(\frac{r+s}{2},\frac{s+v}{2}))]
$$

and, since $\alpha\beta-\gamma^2\ge 0$, the proof follows.

\end{proof}

\vspace{0.5cm}

\section{Applications in Information Theory}

\vspace{0.5cm}

Let
$$
\Omega=\{p=\{p_i\} \ | \ p_i>0, \sum p_i=1\},
$$
be the set of finite or infinite discrete probability
distributions.

\bigskip

One of the most general probability measures which is of
importance in Information Theory is the famous Csisz\'{a}r's
$f$-divergence $C_f(p||q)$ ([5]), defined by

\bigskip

{\bf Definition 1} {\it For a convex function $f: (0,\infty)\to
\mathbb R$, the $f$-divergence measure  is given by
$$
C_f(p||q):=\sum q_i f(p_i/q_i),
$$
where $p, q\in \Omega$.}
\bigskip
By the well known Jensen's inequality for convex functions it
follows that
$$
C_f(p||q)\ge f(1),
$$
with equality  if and only if $p=q$.

\bigskip

 Some important information measures are just particular cases of the Csisz\'{a}r's
 $f$-divergence.

\bigskip

For example,

\bigskip

(a) \ taking $f(x)=x^\alpha, \ \alpha>1$, we obtain the
$\alpha$-order divergence defined by
$$
I_\alpha (p||q):=\sum p_i^\alpha q_i^{1-\alpha};
$$

\bigskip

{\bf Remark 2} \ {\it The above quantity is an argument in
well-known theoretical divergence measures such as Renyi
$\alpha$-order divergence  $I^R_\alpha(p||q)$ or Tsallis
divergence $I^T_\alpha(p||q)$, defined as}
$$
I^R_\alpha(p||q):={1\over \alpha-1}\log I_\alpha(p||q); \ \
I^T_\alpha(p||q):={1\over \alpha-1}(I_\alpha(p||q)-1).
$$

\bigskip

 (b) \ for $f(x)=x\log x$, one obtain the Kullback-Leibler divergence ([3]) defined by
$$
K(p||q):=\sum p_i\log(p_i/q_i);
$$

\bigskip

(c) \ for $f(x)=(\sqrt x -1)^2$, one obtain the Hellinger distance
$$
H^2=H^2(p, q)=H^2(q, p):=\sum (\sqrt p_i-\sqrt q_i)^2;
$$

\bigskip

(d) \ if we choose $f(x)=(x-1)^2$, then we get the
$\chi^2$-distance
$$
\chi^2(p, q):=\sum (p_i-q_i)^2/q_i.
$$

\bigskip

Besides, the quantity

$$
I_{1/2}(p||q)=\sum \sqrt {p_iq_i} :=B(p,q)=B(q,p)=B
$$

is known as {\it Bhattacharya coefficient}. Evidently,

$$
0<B \le 1; \ 2(1-B(p,q))=H^2(p,q).
$$

\bigskip

 We quote now inequalities between those  measures which are
already known in the literature (\cite{c}, \cite{v}).

\bigskip

$$
\chi^2(p, q)\ge H^2(p,q);
$$
$$
K(p||q)\ge H^2(p, q); \eqno (*)
$$
$$
K(p||q)\le \log(1+\chi^2(p, q)).
$$

\bigskip

In particular, $K(p||q)\le \chi^2(p, q)$ (\cite{d}).

\bigskip

The generalized measure $K_s(p||q)$, known as {\it the relative
divergence of type $s$} ([6], [7]), is defined by

\bigskip

$$
K_s(p||q):=\begin{cases}
(\sum p_i^s q_i^{1-s}-1)/ s(s-1) &, s\in \mathbb R/\{0, 1\};\\
                   K(q||p)                         &, s=0;\\
                   K(p||q)                         &, s=1.
                   \end{cases}
                   $$
\bigskip
It include the Hellinger and $\chi^2$ distances as particular
cases.

\bigskip

Indeed,
$$
K_{1/2}(p||q)=4(1-\sum\sqrt {p_i q_i})=2\sum(p_i+q_i-2\sqrt {p_i
q_i})=2H^2(p,q);
$$
$$
K_2(p||q)={1\over 2}(\sum {p_i^2\over q_i}-1)={1\over 2}\sum
{(p_i-q_i)^2\over q_i}={1\over 2}\chi^2(p, q).
$$

\bigskip

It will be proved next that $K_s(p||q)$ is log-convex in $s$ for
$s\in \mathbb R$, wherefrom a whole variety of inequalities
connecting the above mentioned measures arise.

\bigskip

For an application of our results, we shall consider discrete
probability laws in the sequel. The continuous case can be treated
analogously.

\bigskip

For $p,q\in \Omega$, let $\mathcal{A}(q,X)$ be the probability law
with the set of positive discrete random variables $X$ and
assigned probabilities $q$. Choosing $X=\frac{p}{q}$, we get

$$
EX=\sum q_i(\frac{p_i}{q_i})=\sum p_i=1; \ EX^s=\sum
q_i(\frac{p_i}{q_i})^s=\sum p_i^sq_i^{1-s}=I_s(p||q).
$$

Therefore, by continuity, for the probability law
$\mathcal{A}(q,\frac{p}{q})$ we obtain that $\lambda_s=K_s(p||q)$
for each $s\in\mathbb R$.

Hence an important consequence follows.

\vspace{0.5cm}

 {\bf Theorem 6} {\it The relative divergence of type $s$ is logarithmically convex in
$s\in\mathbb R$.}

\vspace{0.5cm}

Noting that $\lambda_{1-s}=K_s(q||p)$, we obtain the following
universal estimations.

\vspace{0.5cm}

{\bf Theorem 7} {\it For each $s\in\mathbb R$, we have

$$
K_s(p||q)K_s(q||p)\ge 4H^4,
$$

and

$$
K_s(p||q)+K_s(q||p)\ge 4H^2,
$$

\vspace{0.5cm}

where $H$, as usual, denotes the Hellinger distance.}

\vspace{0.5cm}

\begin{proof}

Indeed,

$$
K_s(p||q)K_s(q||p)=\lambda_s\lambda_{1-s}\ge\lambda_{1/2}^2=K_{1/2}^2=4H^4.
$$

Furthermore,

$$
K_s(p||q)+K_s(q||p)\ge 2\sqrt{K_s(p||q)K_s(q||p)}\ge 4H^2.
$$

\end{proof}

\vspace{0.5cm}

Denoting the symmetric measures $S_a$ and $P_a$ by

$$
S_a=S_a(p||q)=S_a(q||p):=K_a(p||q)+K_a(q||p)-4H^2;
$$

$$
P_a=P_a(p||q)=P_a(q||p):=K_a(p||q)K_a(q||p)-4H^4,
$$

\vspace{0.5cm}

and applying Theorems 1 and 3, we get

\vspace{0.5cm}

{\bf Theorem 8} {\it Estimations

$$
|S_a-S_b|\le (S_{a+b-1/2}S_{a-b+1/2})^{1/2};
$$

$$
|P_a-P_b|\le (P_{a+b-1/2}P_{a-b+1/2})^{1/2},
$$

hold for each $a,b\in\mathbb R$.}

\vspace{0.5cm}

Further illustration will be given by improving the classical
inequalities $(*)$ for Kullback-Leibler divergence $K(p||q)$.

\vspace{0.5cm}

{\bf Theorem 9} {\it We have

$$
f_1(B, \chi^2)\le K(p||q)\le f_2(B, \chi^2),
$$

where

$$
f_1(B,\chi^2)=-2\log B+6\frac{(1-B^2)^2}{1-B^4+\chi^2(q||p)}
$$

and

$$
f_2(B,
\chi^2)=\log(1+\chi^2(p,q))-\frac{32}{9}\frac{(B\sqrt{1+\chi^2(p,q)}-1)^2}{\chi^2(p,q)}.
$$}

\begin{proof}

Indeed, for the law $\mathcal{A}(p,\frac{p}{q})$ a calculation
shows that

$$
\lambda_0=\log\sum \frac{p_i^2}{q_i}-\sum
p_i\log\frac{p_i}{q_i}=\log(1+\chi^2(p,q))-K(p||q);
$$

$$
\lambda_{-1}=\frac{1}{2}[\sum p_i({p_i}/{q_i})^{-1}-(\sum
{p_i^2}/{q_i})^{-1}]=\frac{1}{2}(1-\frac{1}{1+\chi^2(p,q)});
$$

$$
\lambda_{-1/2}=\frac{4}{3}[\sum p_i({p_i}/{q_i})^{-1/2}-(\sum
{p_i^2}/{q_i})^{-1/2}]=\frac{4}{3}(B(p,q)-\frac{1}{\sqrt{1+\chi^2(p,q)}}).
$$

\bigskip

Now, since $\lambda_0\lambda_{-1}\ge \lambda_{-1/2}^2$, we obtain

$$
K(p||q)\le
\log(1+\chi^2(p,q))-\frac{32}{9}\frac{(B(p,q)\sqrt{1+\chi^2(p,q)}-1)^2}{\chi^2(p,q)},
$$

which is a considerable improvement of the target inequality in
$(*)$.

\bigskip

To improve lower bound of $K(p||q)$ let us consider the law
$\mathcal{A}(p,\sqrt{\frac{q}{p}})$. Since then

$$
\lambda_0=\log(\sum p_i\sqrt{{q_i}/{p_i}})-\sum
p_i\log(\sqrt{q_i/p_i})=\log B(p,q)+\frac{1}{2}K(p||q),
$$

and $\lambda_0\ge 0$, we get a better approximation at once

$$
K(p||q)\ge -2\log B(p,q).
$$

Indeed,

$$
-2\log B(p,q)=-2\log(1-H^2(p,q)/2)\ge H^2(p,q).
$$

\bigskip

We can get further improvement by the inequality
$\lambda_0\lambda_4\ge\lambda_2^2$.

Because,

$$
\lambda_2=\frac{1}{2}(\sum p_i(\sqrt{q_i/p_i})^2-(\sum
p_i\sqrt{q_i/p_i})^2)=\frac{1}{2}(1-B^2(p,q)),
$$

and

$$
\lambda_4=\frac{1}{12}(\sum p_i(\sqrt{q_i/p_i})^4-(\sum
p_i\sqrt{q_i/p_i})^4)=\frac{1}{12}(\chi^2(q||p)+1-B^4(p,q)),
$$

we finally obtain

$$
K(p||q)\ge -2\log
B(p,q)+6\frac{(1-B^2(p,q))^2}{1-B^4(p,q)+\chi^2(q||p)}.
$$
\end{proof}

In this way we get better approximation of Kullback-Leibler
divergence in terms of Bhattacharya coefficient and $\chi^2$-
distance.

\end{document}